\documentstyle{llncs}
\begin{document}

\title{Extensions of the Minimum Cost Homomorphism Problem}

\author{Rustem Takhanov\inst{1}}

\institute{Department of Computer and Information Science, Link{\"o}ping University,\\
SE-581 83 Link{\"o}ping, Sweden\\
takhanov@mail.ru \\}

\maketitle
\sloppy
\begin{abstract}
 Assume $D$ is a finite set and $R$ is a finite set of
functions from $D$ to the natural numbers.
An instance of the
minimum $R$-cost homomorphism problem ($MinHom_R$)
is a set of variables $V$ subject to specified constraints
together with a positive weight $c_{vr}$ for each combination of
$v \in V$ and $r \in R$. The aim is to find a function
$f:V \rightarrow D$ such that $f$ satisfies all constraints
and $\sum_{v \in V} \sum_{r \in R} c_{vr}r(f(v))$ is minimized.

This problem unifies well-known optimization problems
such as the minimum cost homomorphism problem and the
maximum solution problem, and this makes it a computationally
interesting fragment of the valued CSP framework for
optimization problems.
We parameterize $MinHom_R\left(\Gamma\right)$ by
{\em constraint languages}, i.e. sets $\Gamma$ of
relations that are allowed in constraints.
A constraint language is called {\em conservative} if every
unary relation is a member of it; such constraint languages
play an important role in understanding the structure of
constraint problems.
The dichotomy conjecture for $MinHom_R$ is the following statement: if
$\Gamma$ is a constraint language, then
$MinHom_R\left(\Gamma\right)$ is either polynomial-time solvable or
NP-complete.
For $MinHom$ the dichotomy result has been recently obtained [Takhanov, STACS, 2010]
and the goal of this paper is to expand this result to the case of $MinHom_R$ with conservative constraint language.
For arbitrary $R$ this problem is still open, but assuming certain restrictions on $R$ we prove a dichotomy.
As a consequence of this result we obtain a dichotomy for the conservative maximum solution problem.
\end{abstract}
\section{Introduction}
Constraint satisfaction problems ($CSP$) and
valued constraint satisfaction problems ($VCSP$) are natural ways of formalizing a large number of
computational problems arising in combinatorial optimization, artificial intelligence,
and database theory. $CSP$ has the following two equivalent formulations:
(1) to find an assignment of values to a given set of variables, subject to constraints on the values that can be
assigned simultaneously to specified subsets of variables, and
(2) to find a homomorphism between two finite relational structures $A$ and $B$.
$VCSP$ is a ``soft'' version of $CSP$ where constraint relations are replaced by functions
from set of tuples to some totally ordered set with addition operation (for example, rational numbers).
A solution is defined as an assignment to variables that maximize a functional which is equal to
a sum of constraint functions applied to corresponding variables.
Applications of $CSP$s and $VCSP$s arise in the propositional logic, database and graph theory, scheduling, biology and many other areas.
$CSP$ and its subproblems has been intensively studied by computer scientists and mathematicians since the 70s, and
recently attention has been paid to its modifications such as $VCSP$.
Considerable attention has been given to the case where the
constraints are restricted to a given finite set of relations $\Gamma$, called a constraint language \cite{bulatov,feder,jeavons,schaefer}.
For example, when $\Gamma$ is a constraint language over the boolean set $\{0, 1\}$ with four ternary predicates $x\vee y\vee z$,
$\overline{x}\vee y\vee z$, $\overline{x}\vee \overline{y}\vee z$,
$\overline{x}\vee \overline{y}\vee \overline{z}$ we obtain 3-SAT.
For every constraint language $\Gamma$, it has been conjectured that $CSP\left( \Gamma \right)$ is either in P or NP-complete \cite{feder}.
An analogous situation appears in $VCSP$ where the constraint language is defined as a set of ``soft'' predicates.

We believe that problems like minimum cost homomorphism problem ($MinHom$) has an intermediate position between $CSP$s and $VCSP$s which
makes their structure important for understanding the relationship between ``hard'' and ``soft'' constraints in optimization.
In the minimum cost homomorphism problem, we are given variables subject to constraints and, additionally, costs on variable/value pairs.
Now, the task is not just to find any satisfying assignment to the variables, but one that minimizes the total cost.
In the context of $VCSP$ this is equivalent to addition of ``soft'' constraints equal to characteristic functions of one element sets. We will consider a weighted version of this problem.

\begin{definition}
Suppose we are given a finite domain set $A$, a finite constraint language $\Gamma \subseteq \bigcup\limits_{k = 1}^\infty  {{2^{{A^k}}}} $ and a finite set of functions $R \subseteq \left\{r:A\rightarrow \bbbn\right\}$.
Denote by
$MinHom_{R}\left( \Gamma \right)$
the following minimization task:

\noindent {\bf Instance:} A triple $\left(V, {\sf C}, W\right)$ where
\begin{itemize}
\item $V$ is a set of variables;
\item ${\sf C}$ is a set of constraints, where each constraint $C \in {\sf C}$ is a pair $(s,\rho)$,
such that
\begin{itemize}
\item $s = \left(v_1,\dots, v_m\right)$ is a tuple of variables of length $m$, called the
constraint scope;
\item $\rho$ is an element of $\Gamma$ with arity $m$,
called the constraint relation.
\end{itemize}
\item Weights $w_{vr}\in \bbbn, v\in V, r\in R$.
\end{itemize}

\noindent {\bf Solution:} A function $f$ from $V$ to $A$,
such that, for each variable $v\in V$, $f(v)\in A$, and for each constraint
$(s,\rho)\in {\sf C}$, with $s = \left(v_1,\dots, v_m\right)$, the tuple $\left(f(v_1),\dots,f(v_m)\right)$ belongs to $\rho$.

\noindent {\bf Measure:} $\sum\limits_{v\in V} \sum\limits_{r\in R}{w_{vr}r\left( f\left( v \right) \right)}$.
\end{definition}

\begin{definition}
For $R^{*} = \{e_i|i\in A\}$, $MinHom_{R^{*}}\left( \Gamma \right)$ is called {\em minimum cost homomorphism problem}
where $e_i: A\rightarrow \bbbn$ denotes a characteristic function of $\{i\}\subseteq A$.
\end{definition}

We will write $MinHom$ instead of $MinHom_{R^{*}}$ for short. $MinHom$ has applications in defence logistics \cite{gutin2} and machine learning \cite{daniels}.
Complete classification of constraint languages $\Gamma$ for which
$MinHom\left( \Gamma \right)$ is polynomial-time solvable has recently been obtained in \cite{takhan}.
The question for which directed graphs $H$ the problem
$MinHom\left( \left\{ H \right\} \right)$ is polynomial-time solvable was considered in \cite{gupta,gutin0,gutin1,gutin2,gutin3}.
Maximum Solution Problem ($MaxSol$), which is defined analogously to $MinHom$, but with $A\subseteq  \bbbn$ and a functional of the
form $\sum\limits_{v\in V} {w_{v}f\left( v \right)}$ to maximize, was investigated in a series of papers \cite{jonnson0,jonnson1,jonnson2}.
It is easy to see that if $n = \max\limits_{s\in A} s + 1$ and $R = \left\{n-x\right\}$, then $MinHom_{R}\left( \Gamma \right) = MaxSol\left( \Gamma \right)$.
In this paper, we will assume that a constraint language $\Gamma$ contains all unary predicates
over a domain set $A$ and approach the problem of characterizing the complexity of $MinHom_{R}\left( \Gamma \right)$
in its most general form by algebraic methods.
When $R$ satisfies certain conditions, we obtain a dichotomy for $MinHom_{R}\left( \Gamma \right)$,
i.e., if $MinHom_{R}\left( \Gamma \right)$ is not polynomial-time solvable, then it is NP-hard. As a consequence, we
obtain a dichotomy for conservative $MaxSol$.

In Section 2, we present some preliminaries together with results connecting the complexity of $MinHom_R$ with conservative algebras. The main dichotomy theorem is stated in Section 3 and its proof is divided into several parts which can be found in Sections 4-6. Finally, in Section 7 we present directions for future research.

\section{Algebraic structure of tractable constraint languages}
Recall that an
optimization problem $A$ is called NP-hard if some NP-complete language can be
recognized in polynomial time with the aid of an oracle for $A$. We assume that $P \ne NP$.

\begin{definition}
Suppose we are given a finite set $A$ and a constraint language
$\Gamma \subseteq \bigcup\limits_{k = 1}^\infty  {{2^{{A^k}}}} $. The language $\Gamma$ is
said to be {\em $R$-tractable} if, for every finite subset
$\Gamma' \subseteq \Gamma$, the task $MinHom_{R}\left( \Gamma' \right)$ is polynomial-time solvable, and $\Gamma$ is
called {\em $R$-NP-hard} if there is a finite subset $\Gamma' \subseteq \Gamma$, such that
the task $MinHom_{R}\left( \Gamma' \right)$ is NP-hard.
\end{definition}

First, we will state some standard definitions from universal algebra.

\begin{definition}
Let $\rho  \subseteq A^m $ and $f:A^n  \to A$.
We say that the function (operation) $f$ {\em preserves} the predicate $\rho $ if,
for every $\left( {x_1^i ,\dots,x_m^i } \right) \in \rho, 1 \leq i \leq n$, we have that $\left( {f\left( {x_1^1 ,\dots,x_1^n
} \right),\dots,f\left( {x_m^1 ,\dots,x_m^n } \right)} \right) \in
\rho $.
\end{definition}

For a constraint language $\Gamma$, let $Pol\left( \Gamma \right)$ denote the
set of operations preserving all predicates in $\Gamma$. Throughout the paper, we let $A$ denote a finite domain and $\Gamma$ a constraint language over $A$. We assume the domain $A$ to be finite.

\begin{definition}
A constraint language $\Gamma$ is called {\em a relational
clone} if it contains every predicate expressible by a first-order formula involving only

a) predicates from $\Gamma \cup \left\{=^{A}\right\}$;

b) conjunction; and

c) existential quantification.
\end{definition}

First-order formulas involving only conjunction and existential quantification are
often called {\em primitive positive (pp) formulas}. For a given constraint language $\Gamma$, the set of all predicates that can be
described by pp-formulas over $\Gamma$ is called the {\em closure} of $\Gamma$ and is denoted by $\langle \Gamma \rangle$.

For a set of
operations $F$ on $A$, let $Inv\left( F \right)$ denote the set of
predicates preserved under the operations of $F$. Obviously, $Inv\left( F \right)$ is a relational clone.
The next result is well-known \cite{bodnarchuk,geiger}.

\begin{theorem}
\label{kuznetsov}
For a constraint language $\Gamma$ over a finite set $A$, $\langle \Gamma \rangle = Inv\left( Pol\left( \Gamma \right) \right)$.
\end{theorem}

Theorem \ref{kuznetsov} tells us that the Galois closure of a constraint language $\Gamma$
is equal to the set of all predicates that can be obtained
via pp-formulas from the predicates in $\Gamma$.
We will omit the proof of the following standard theorem.

\begin{theorem}
\label{standard}
For any finite constraint language $\Gamma$, finite $\Gamma'\subseteq \langle \Gamma \rangle$ and finite $R \subseteq \left\{r:A\rightarrow \bbbn\right\}$ there is a polynomial time reduction from $MinHom_{R}\left( \Gamma' \right)$ to $MinHom_{R}\left( \Gamma \right)$.
\end{theorem}

The previous theorem tells us that the complexity of $MinHom_{R}\left( \Gamma \right)$ is basically
determined by $Inv\left( Pol\left( \Gamma \right) \right)$, i.e., by $Pol\left( \Gamma \right)$.
That is why we will be concerned with the classification of sets of operations $F$ for which $Inv\left( {F} \right)$ is a tractable constraint language.

\begin{definition}
An {\em algebra} is an ordered pair ${\cal A} = \left(A, F\right)$ such that $A$ is
a nonempty set (called a universe) and $F$ is a family of finitary operations on $A$. An
algebra with a finite universe is referred to as a finite algebra.
\end{definition}

\begin{definition}
An algebra ${\cal A} = \left(A, F\right)$ is called {\em $R$-tractable} if $Inv(F)$ is a $R$-tractable constraint language and ${\cal A}$ is called {\em $R$-NP-hard} if $Inv(F)$ is an $R$-NP-hard constraint language.
\end{definition}

For $B\subseteq A$, define $R_B = \left\{f|_{B}|f\in R\right\}$, where $f|_B$ is a restriction of $f$ on a set $B$.
We will use the term MinHom-tractable (NP-hard) instead of $\{e_i|i\in A\}$-tractable (NP-hard) and, in case $A=\left\{0,1\right\}$, the term min-tractable (NP-hard) instead of $\{x\}$-tractable (NP-hard).

We only need to consider a very special type of algebras, so called {\em conservative} algebras.

\begin{definition}
An algebra ${\cal A} = \left(A, F\right)$ is called {\em conservative} if for every operation $f \in F$ we have that $f\left( {x_1 ,\dots,x_n } \right)
\in \left\{ {x_1 ,\dots,x_n } \right\}$.
\end{definition}

Since we assume that $\Gamma$ is a constraint language with all unary relations over the domain set $A$, then ${\cal A} = \left(A, Pol\left( \Gamma \right)\right)$ is conservative.
Besides conservativeness of constraint languages we will make some additional restrictions on function sets $R$.

\begin{definition}
Suppose we are given a finite set of functions $R \subseteq \left\{r:A\rightarrow \bbbn\right\}$.
Denote by $G\left( R \right) = \left( A, E\left( R \right) \right)$ a directed graph with a set of vertices $A$ and an edge set $E\left( R \right) = \left\{(a,b)| \exists r\in R\,\, r ( a ) > r ( b )\right\}$. The $UG\left( R \right)$ is the graph $G\left( R \right)$ with all edges considered as undirected. We will call $G$ {\em a preference graph} and $UG$ {\em an undirected preference graph}.
\end{definition}

In the sequel, we will assume that a graph $UG$ is complete. It is easy to see that $UG\left( \{e_i|i\in A\} \right)$, $UG\left( \left\{n-x\right\} \right)$ are complete and our results can be applied to $MinHom$ and $MaxSol$.

\section{Boolean case and the necessary local conditions}
The first step to understand the structure of $R$-tractable
algebras is to understand the boolean case.
Well-known structure of boolean clones \cite{post} helps us to prove the following theorem.

\begin{theorem}
\label{by_post}
A boolean clone $H$ is MinHom-tractable if either $\left\{x\wedge y, x\vee y\right\} \subseteq H$ or $\left\{\left( {x \wedge \overline y } \right) \vee \left( {\overline y  \wedge z} \right) \vee \left( {x \wedge z} \right)\right\} \subseteq H$, where $\overline{x}, x\wedge y, x\vee y$ denote negation, conjunction and disjunction.
Otherwise, $H$ is MinHom-NP-hard. A conservative boolean clone $H$ is min-tractable if either $\left\{x\wedge y\right\} \subseteq H$ or $\left\{\left( {x \wedge \overline y } \right) \vee \left( {\overline y  \wedge z} \right) \vee \left( {x \wedge z} \right)\right\} \subseteq H$.
Otherwise, $H$ is min-NP-hard.
\end{theorem}

In the proof of Theorem \ref{by_post} we will need the following definition.

\begin{definition}
A constraint language $\Gamma$ over $\left\{0, 1\right\}$ is called a
{\em MinHom(min)-maximal} constraint language if it is conservative MinHom(min)-tractable and is not
contained in any other conservative MinHom(min)-tractable languages.
\end{definition}

We identify all MinHom(min)-maximal constraint languages using Post`s
classification \cite{post}.
Via Theorems \ref{kuznetsov},\ref{standard} we conclude that every MinHom(min)-maximal constraint language corresponds to some conservative
functional clone.
In the case $A = \left\{ {0,1} \right\}$, there is a countable number
of conservative clones: we list them below according to the table on page 76 \cite{marchenkov}.

\begin{lemma}
\label{30}
The relational clones
$Inv\left( {M_{01} } \right)$ and $Inv\left( {S_{01} } \right)$
are MinHom-maximal constraint languages. Every other constraint language
given in the table that is not contained in any of these two
is MinHom-NP-hard.
The relational clones
$Inv\left( {K_{01} } \right)$ and $Inv\left( {S_{01} } \right)$
are min-maximal constraint languages. Every other constraint language given in the table that is not contained in any of these two
is min-NP-hard.
\end{lemma}

\begin{proof}
 For every row, the closure of the
predicates given is equal to the set of all predicates preserved
under the functions of the corresponding clone.
\[
\begin{array}{*{20}c}
   {T_{01} } & {x = 0,x = 1}  \\
   {M_{01} } & {x = 0,x = 1,x_1  \le x_2 }  \\
   {S_{01} } & {x = 0,x_1  \ne x_2 }  \\
   {SM} & {x_1  \ne x_2 ,x_1  \le x_2 }  \\
   {L_{01} } & {x = 1,x_1  \oplus x_2  \oplus x_3  = 0}  \\
   {U_{01} } & {x = 0,x = 1,x_1  = x_2  \vee x_1  = x_3 }  \\
   {K_{01} } & {x = 0,x = 1,x_1  = x_2 x_3 }  \\
   {D_{01} } & {x = 0,x = 1,x_1  = x_2  \vee x_3 }  \\
   {I_1^m } & {x = 1,x_1 x_2 \dots x_m  = 0}  \\
   {MI_1^m } & {x = 1,x_1  \le x_2 ,x_1 x_2 \dots x_m  = 0}  \\
   {O_0^m } & {x = 0,x_1  \vee x_2  \vee \dots \vee x_m  = 1}  \\
   {MO_0^m } & {x = 0,x_1  \le x_2 ,x_1  \vee x_2  \vee \dots \vee x_m  = 1}  \\
\end{array}
\]
where $x\oplus y = x+y\left(mod\rm{\,\,}2\right)$.

The class $Inv\left( {T_{01} } \right)$ is MinHom(min)-tractable, since it
contains only two simple unary predicates
$\{0\}$ and $\{1\}$. As we will see later, it cannot be
MinHom(min)-maximal since it is included in other MinHom(min)-tractable constraint languages.

Let us prove that $Inv\left( {M_{01} }
\right)$ and $Inv\left( {S_{01} } \right)$ are MinHom-tractable. By Theorem \ref{standard}, it is
equivalent to polynomial solvability of $MinHom\left( {\left\{
{\left\{ 0 \right\},\left\{ 1 \right\},\left\{ {\left( {x_1 ,x_2
} \right)|x_1  \le x_2 } \right\}} \right\}} \right)$ and $MinHom\left( \left\{ \left\{ 0
\right\},\left\{ \left( {x_1 ,x_2 } \right)|x_1  \ne x_2
\right\}\right\} \right)$, because the
classes $Inv\left( {M_{01} } \right)$ and $Inv\left( {S_{01} } \right)$ are the closures of those sets.
We will skip the proof since it can be found in \cite{jonnson}(in this paper boolean $MinHom$ called Max AW Ones).
It is easy to see that MinHom-tractability implies min-tractability of those classes.

Let us show that all the classes in the table, except $Inv\left( {M_{01} } \right)$,
$Inv\left( {S_{01} } \right)$ and $Inv\left( {T_{01} } \right)$,
are MinHom-NP-hard, and all the classes, except $Inv\left( {M_{01} } \right)$,
$Inv\left( {S_{01} } \right)$, $Inv\left( {T_{01} } \right)$, $Inv\left( {K_{01} } \right)$ and $Inv\left( {MI_1^{2}} \right), \dots, Inv\left( {MI_1^{\infty} } \right)$, where $MI_1^{\infty} = \bigcup\limits_{m=1}^{\infty}MI_1^{m}$, are min-NP-hard.
Since,
\[
\begin{array}{l}
 x_1  \vee x_2  = \exists x_3 \left[ {x_1  \ne x_3 } \right]\wedge \left[ {x_3  \le x_2 } \right] \\
 x_1  \vee x_2  = \exists x_3 \left[ {x_3  = 1} \right]\wedge \left[ {x_3  = x_1  \vee x_3  = x_2 } \right] \\
 \overline {x_1 }  \vee \overline {x_2 }  = \exists x_3 \left[ {x_3  = 0} \right]\wedge \left[ {x_3  = x_1 x_2 } \right] \\
 x_1  \vee x_2  = \exists x_3 \left[ {x_3  = 1} \right]\wedge \left[ {x_3  = x_1  \vee x_2 } \right] \\
 \overline {x_1 }  \vee \overline {x_2 }  = \exists x_3 \dots x_m \left[ {x_1 x_2 \dots x_m  = 0} \right]\wedge \left[ {x_2  = x_3 } \right]\wedge \dots\wedge \left[ {x_{m - 1}  = x_m } \right] \\
 x_1  \vee x_2  = \exists x_3 \dots x_m \left[ {x_1  \vee x_2  \vee \dots \vee x_m  = 1} \right]\wedge \left[ {x_2  = x_3 } \right]\wedge \dots\wedge \left[ {x_{m - 1}  = x_m } \right] \\
 \end{array}
\]
we see that $\left\{ {\left( {x_1 ,x_2 } \right)|x_1  \vee x_2 }
\right\} \in Inv\left( {SM} \right)$, $Inv\left( {U_{01} }
\right),Inv\left( {D_{01} } \right)$, $Inv\left( {O_0^m }
\right),Inv\left( {MO_0^m } \right)$ and $\left\{ {\left( {x_1
,x_2 } \right)|\overline {x_1 }  \vee \overline {x_2 } } \right\}
\in Inv\left( {K_{01} } \right)$, $Inv\left( {I_1^m } \right)$, $
Inv\left( {MI_1^m } \right)$.

We first prove that
$MinHom_{\{x\}}\left( {\left\{ {\left\{
{\left( {x_1 ,x_2 } \right)|x_1  \vee x_2 } \right\}} \right\}}
\right)$ is NP-hard. Suppose an instance of this problem consists of an undirected graph $G=\left(V,E\right)$ where each vertex is considered as a variable. For each pair of variables $(u,v)\in E$, we require their assignments to satisfy $u = 1$ or $v = 1$. It is easy to see
that for any such assignment $f$, the set $\{x|f(x)=0\}$ is independent
in the graph $G$. Furthermore, for any independent set $S$ in the graph $G$,
$g(x) = [x\notin S]$ is a satisfying assignment.
If we define $w_{i} = 1$ for $i\in V$, then
$MinHom_{\{x\}}$ is equivalent to finding a maximum independent set.
This implies that $MinHom_{\{x\}}\left( {\left\{ {\left\{
{\left( {x_1 ,x_2 } \right)|x_1  \vee x_2 } \right\}} \right\}}
\right)$ is NP-hard, since finding independent sets of maximal size is an NP-hard problem.

Therefore, $Inv\left( {SM} \right)$,
$Inv\left( {U_{01} } \right)$, $Inv\left( {D_{01} } \right)$,
$Inv\left( {O_0^m } \right)$, $Inv\left( {MO_0^m } \right)$ are min-NP-hard, and, consequently, MinHom-NP-hard.

Classes $Inv\left( {K_{01} } \right)$, $Inv\left( {I_1^m } \right)$,
$Inv\left( {MI_1^m } \right)$ are MinHom-NP-hard also since $MinHom\left( {\left\{ {\left\{
{\left( {x_1 ,x_2 } \right)|x_1  \vee x_2 } \right\}} \right\}}
\right)$ and $MinHom\left( {\left\{ {\left\{ {\left( {x_1 ,x_2 }
\right)|\overline {x_1 }  \vee \overline {x_2 } } \right\}}
\right\}} \right)$ are equivalent. Let us show that these classes are min-tractable.
It is easy to see that $Inv\left( {K_{01} } \right)$ contains $Inv\left( {I_1^m } \right)$,
$Inv\left( {MI_1^m } \right)$ and $\wedge\in K_{01}$. Indeed, any constraint satisfaction problem
with predicates from $Inv\left( {K_{01} } \right)$ can be solved by local 1-consistency algorithm and a solution is an
assignment of every variable to a minimum of its allowed values. Obviously, the same algorithm solves $MinHom_{\{x\}}\left( {Inv\left( {K_{01} } \right) } \right)$.

It remains to prove min-NP-hardness of $Inv\left( {L_{01} } \right)$ which will also show its MinHom-NP-hardness.
First we will show that using an algorithm for $MinHom_{\{x\}}\left( {\left\{ {\left( {x_1 ,x_2 ,x_3 } \right)|x_1  \oplus x_2 \oplus x_3  = 1}
\right\}} \right)$ as an oracle, we can solve Max-CUT in
polynomial time. Since we have that $x_1  \oplus x_2 \oplus x_3  = 1 \Leftrightarrow \exists y,z\,\,x_1  \oplus x_2 \oplus y  = 0 \& y  \oplus x_3 \oplus z  = 0 \& z = 1$, we will conclude that $Inv\left( {L_{01} } \right)$ is min-NP-hard.

Let $G = \left( {V,E} \right)$ be a graph and introduce
variables $x_{ij} ,y_i ,y_j ,i,j \in V$. A system of equations $x_{ij} \oplus y_i  \oplus y_j  = 1,i,j \in V$ can be viewed as an
instance of $MinHom_{\{x\}}\left( {\left\{ {\left( {x_1 ,x_2 ,x_3 } \right)|x_1  \oplus x_2 \oplus x_3  = 1}
\right\}} \right)$. It is easy to see that arbitrary boolean vector
$\overline y = \left( {y_1 ,\dots,y_{\left| V \right|} } \right)$ defines a single solution $x_{ij} = y_i \oplus y_j  \oplus 1,i,j \in V$ of the system. Vector $\overline y $ can be considered as the cut
$\left\{ {i|y_i  = 1} \right\} \subseteq V$ and the value $\sum\limits_{ij}(1-x_{ij})$ is equal to the doubled cost of the cut. Then Max-CUT is polynomially reduced to minimizing $\sum\limits_{ij}x_{ij}$ which is equivalent to $MinHom_{\{x\}}$.

Only two classes $Inv\left( {M_{01} } \right)$ and
$Inv\left( {S_{01} } \right)$ are left as candidates for MinHom-maximality. Since they are not
included in each other, they are both maximal. The same argument shows that $Inv\left( {K_{01} } \right)$ and
$Inv\left( {S_{01} } \right)$ are min-maximal.
\end{proof}

\noindent {\bf Proof of Theorem 12.}
Obviously, $\left\{\wedge, \vee\right\}$ is a basis of $M_{01}$, $\left\{\wedge\right\}$ is
a basis of $K_{01}$
and $\left\{\left( {x \wedge \overline y } \right) \vee \left( {\overline y  \wedge z} \right) \vee \left( {x \wedge z} \right)\right\}$ is a basis of $S_{01}$. By Lemma \ref{30} we conclude the statement of the theorem.
\qed

\begin{definition}
We call a pair of vertices $\{a,b\}$ in $G\left( R \right)$ {\em a MinHom-pair} if they have arcs in both directions and we call $\{a,b\}$ {\em a min-pair} if they have an arc in one direction only.
\end{definition}

Suppose $f\in F$. By $\mathop  \downarrow \limits_b^a f$, we mean $a \ne b$ and $f \left( {a,b} \right) = f \left( {b,a} \right) = b$.
Guided by Theorem \ref{by_post} we obtain the following definition.

\begin{definition}
Let $F$ be a conservative functional clone over $A$ and $R \subseteq \left\{r:A\rightarrow \bbbn\right\}$. We say that $F$ satisfies the {\em necessary local conditions for $R$} if and only if
\begin{itemize}
\item[(1)]
for every MinHom-pair $\{a,b\}$ of $G\left( R \right)$, either
\begin{itemize}
\item[(1.a)]
there exists $f_1, f_2\in F$ s.t. $\mathop  \downarrow \limits_b^a f_1$ and $\mathop  \uparrow \limits_b^a f_2$; or
\item[(1.b)]
there exists $f\in F$ s.t. $f|_{\{a,b\}} \left( {x,x,y} \right) = f|_{\{a,b\}} \left( {y,x,x} \right) = f|_{\{a,b\}} \left( {y,x,y} \right) = y$,
\end{itemize}
\item[(2)]
for every min-pair $\{a,b\}$ of $G\left( R \right)$ such that $(a,b)\in E\left( R \right)$, either
\begin{itemize}
\item[(2.a)]
there exists $f\in F$ s.t. $\mathop  \downarrow \limits_b^a f$; or
\item[(2.b)]
there exists $f\in F$ s.t. $f|_{\{a,b\}} \left( {x,x,y} \right) = f|_{\{a,b\}} \left( {y,x,x} \right) = f|_{\{a,b\}} \left( {y,x,y} \right) = y$.
\end{itemize}
\end{itemize}
\end{definition}

\begin{theorem}
Suppose $F$ is a conservative functional clone. If $F$ is $R$-tractable and $UG\left( R \right)$ is complete, then it satisfies the necessary local conditions for $R$. If $F$ does not satisfy the necessary local conditions for $R$, then it is $R$-NP-hard.
\end{theorem}

As in case of $MinHom$, the necessary local conditions are not sufficient for
$R$-tractability of a conservative clone. Let $M = \left\{ {B|B \subseteq A, \left| B \right| = 2,F|_B {\rm{\,\,contains\,\,2\,\,different\,\,binary\,\,commutative\,\,functions}}} \right\}$, $M^o  = \left\{ {\left( {a,b} \right)|\left\{ {a,b} \right\} \in M} \right\}$ and $\overline M  = \left\{ {B|B \subseteq A, \left| B \right| = 2} \right\}\backslash M$.

Introduce an undirected graph without loops $T_F^R = \left( {M^o\cap E\left( R \right) ,P} \right)$ where $P = \left\{ {\left\langle {\left( {a,b} \right),\left( {c,d} \right)} \right\rangle |\left( {a,b} \right),\left( {c,d} \right) \in M^o\cap E\left( R \right) ,{\rm{\,\,there\,\,is\,\,no\,\,}}f  \in F:\mathop  \downarrow \limits_b^a \mathop  \downarrow \limits_d^c f } \right\}$.

\begin{theorem}
\label{main}
Suppose $F$ satisfy the necessary local conditions for $R$ and $UG\left( R \right)$ is complete. If the graph $T_F^R = \left( {M^o\cap E\left( R \right) ,P} \right)$ is bipartite, then $F$
is $R$-tractable. Otherwise, $F$ is $R$-NP-hard.
\end{theorem}

A proof for NP-hard case of Theorem \ref{main} will be omitted since it is basically the same as in case of $MinHom$ \cite{takhan}.

\section{Multi-sorted MinHom and its tractable case}
\begin{definition}
For any collection of sets ${\sf A} = \left\{A_i| i\in I\right\}$, and any list of
indices $i_1,\dots, i_m\in I$, a subset $\rho$ of $A_{i_1}\times\dots\times A_{i_m}$, together with
the list $\left(i_1,\dots, i_m\right)$, will be called {\em a multi-sorted relation over ${\sf A}$ with arity $m$
and signature $\left(i_1,\dots, i_m\right)$}. For any such relation $\rho$, the signature of $\rho$ will be
denoted $\sigma (\rho)$.
\end{definition}

\begin{definition}
Let $\Gamma$ be a set of multi-sorted relations over a collection of sets
${\sf A} = \left\{A_i| i\in I\right\}$. The multi-sorted $MinHom$ problem over $\Gamma$,
denoted $MMinHom(\Gamma)$, is defined to be the minimization problem with

\noindent {\bf Instance:} A quadruple $\left(V, \delta, {\sf C}, W\right)$ where
\begin{itemize}
\item $V$ is a set of variables;
\item $\delta$ is a mapping from $V$ to $I$, called the domain function;
\item ${\sf C}$ is a set of constraints, where each constraint $C \in {\sf C}$ is a pair $(s,\rho)$,
such that
\begin{itemize}
\item $s = \left(v_1,\dots, v_m\right)$ is a tuple of variables of length $m$, called the
constraint scope;
\item $\rho$ is an element of $\Gamma$ with arity $m$ and signature $\left(\delta(v_1),\dots,\delta(v_m)\right)$,
called the constraint relation.
\end{itemize}
\item Weights $w_{va}\in \bbbn, v\in V, a\in A_{\delta(v)}$.
\end{itemize}

\noindent {\bf Solution:} A function $f$ from $V$ to $\bigcup \limits_{i\in I} A_i$,
such that, for each variable $v\in V$, $f(v)\in A_{\delta(v)}$, and for each constraint
$(s,\rho)\in {\sf C}$, with $s = \left(v_1,\dots, v_m\right)$, the tuple $\left(f(v_1),\dots,f(v_m)\right)$ belongs to $\rho$.

\noindent {\bf Measure:} $\sum\limits_{v\in V} {w_{vf\left( v \right)}}$.
\end{definition}

We can consider any multi-sorted relation $\rho$ over ${\sf A} = \left\{A_i| i\in I\right\}$ as an ordinary relation $\rho^{\sf A}$ over a set $\bigcup \limits_{i\in I}A_i$ where $A_i, i\in I$ are considered to be disjoint. If $\Gamma$ is a set of multi-sorted relations over ${\sf A} = \left\{A_i| i\in I\right\}$, then $\Gamma^{\sf A}$ denotes a set of relations of $\Gamma$ considered as relations over $\bigcup \limits_{i\in I}A_i$. It is easy to see that $MMinHom(\Gamma)$ is equivalent to $MinHom(\Gamma^{\sf A})$.

\begin{definition}
A set of multi-sorted relations over ${\sf A}$, $\Gamma$, is said to be {\em MinHom-tractable}, if $\Gamma^{\sf A}$ is MinHom-tractable.
A set of multi-sorted relations over ${\sf A}$, $\Gamma$, is said to be {\em MinHom-NP-complete}, if $\Gamma^{\sf A}$
is MinHom-NP-complete.
\end{definition}

\begin{definition}
Let ${\sf A}$ be a collection of sets. {\em An $n$-ary multi-sorted operation
$t$ on ${\sf A}$} is defined by a collection of interpretations $\left\{t^A | A \in {\sf A}\right\}$, where each
$t^A$ is an $n$-ary operation on the corresponding set $A$. The multi-sorted operation
$t$ on $A$ is said to be {\em a polymorphism} of a multi-sorted relation $\rho$ over ${\sf A}$ with
signature $(\delta(1),\dots,\delta(m))$ if, for any $(a_{11},\dots,a_{m1}),\dots, (a_{1n},\dots,a_{mn}) \in \rho$, we
have
\[
t\left( {\begin{array}{*{20}{c}}
   {{a_{11}}} &  \cdots  & {{a_{1n}}}  \\
    \vdots  & {} &  \vdots   \\
   {{a_{m1}}} &  \cdots  & {{a_{mn}}}  \\
\end{array}} \right) = \left( {\begin{array}{*{20}{c}}
   {{t^{\delta \left( 1 \right)}}\left( {{a_{11}}, \ldots ,{a_{1n}}} \right)}  \\
    \vdots   \\
   {{t^{\delta \left( m \right)}}\left( {{a_{m1}}, \ldots ,{a_{mn}}} \right)}  \\
\end{array}} \right) \in \rho
\]
\end{definition}

For any given set of multi-sorted relations $\Gamma$, $MPol\left(\Gamma\right)$ denotes the set of multi-sorted
operations which are polymorphisms of every relation in $\Gamma$.

\begin{definition}
Suppose a set of operations $H$ over $D$ is conservative and $B\subseteq \left\{ {\left\{ {x,y} \right\}|x,y \in D ,x \ne y} \right\}$. A pair of binary operations $\phi, \psi \in H$ is called a {\em tournament pair} on $B$,
if $\forall \left\{ {x,y} \right\} \in B {\rm{\,\,}}\phi \left( {x,y} \right) = \phi \left( {y,x} \right),\psi \left( {x,y} \right) = \psi \left( {y,x} \right), {\phi \left( {x,y} \right)\ne\psi \left( {x,y} \right)}$
and for arbitrary $\left\{ {x,y} \right\} \in \overline{B} $, $\phi \left( {x,y} \right) = x,\psi \left( {x,y} \right) = x$.
An operation  $m \in H$ is called {\em arithmetical} on $B$, if
$\forall \left\{ {x,y} \right\} \in B {\rm{\,\,}}m \left( {x,x,y} \right) = m \left( {y,x,x} \right) = m \left( {y,x,y} \right) = y$.
\end{definition}

The following theorem is a simple consequence of the main result of \cite{takhan}.

\begin{theorem}
\label{minhom}
Let $\Gamma$ be a constraint language over $A$ containing all unary relations and $B\subseteq \left\{\{a,b\}|a,b\in A, a\ne b\right\}$.
If $Pol\left(\Gamma\right)$ contains operations $\phi, \psi, m$ such that
\begin{itemize}
\item $\phi, \psi$ is a tournament pair on $B$,
\item $m$ is arithmetical on $\overline{B}$,
\end{itemize}
then $\Gamma$ is MinHom-tractable.
\end{theorem}

The following theorem is a generalization of the previous one.

\begin{theorem}
\label{multi}
Let $\Gamma$ be a set of multi-sorted relations over a collection of finite
sets ${\sf A} = \left\{A_1,\dots,A_n\right\}$ containing all unary multi-sorted relations.
Assume that $B_i\subseteq \left\{\{a,b\}|a,b\in A_i, a\ne b\right\}$.
If $MPol\left(\Gamma\right)$ contains a multi-sorted operations $\phi, \psi, m$ such that
\begin{itemize}
\item $\phi^{A_i}, \psi^{A_i}$ is a tournament pair on $B_i$,
\item $m^{A_i}$ is arithmetical on $\overline{B_i}$,
\end{itemize}
then $\Gamma ^{\sf A}$ is MinHom-tractable.
\end{theorem}

\begin{proof}
Denote $G = \bigcup\limits_{i=1}^{n}{A_i}$. It is easy to see that we can define operations
$\phi',\psi' : G^2\rightarrow G$ and $m':G^3\rightarrow G$
such that $\phi'|_{A_i} = \phi^{A_i}$, $\psi'|_{A_i} = \psi^{A_i}$, $m'|_{A_i} = m^{A_i}$ and $\phi', \psi'$ is a tournament pair on
$\bigcup\limits_{i=1}^{n}{B_i} \cup \left\{\{a,b\}|a\in A_i,b\in A_j, i\ne j\right\}$. Then $\phi',\psi', m' \in Pol\left(\Gamma^{\sf A}\right)$
which by Theorem \ref{minhom} means that $\Gamma ^{\sf A}$ is MinHom-tractable.
\end{proof}

\section{Structure of $R$-tractable algebras}
\begin{definition}
Suppose a set of operations $H$ over $D$ is conservative and $O\subseteq \left\{ {\left( {x,y} \right)|x,y \in D ,x \ne y} \right\}$. A pair of binary operations $\phi, \psi \in H$ is called a {\em weak tournament pair} on $O$,
if
\begin{itemize}
\item[(1)]
$\forall\,\,a,b {\rm{\,\,such\,\,that\,\,}}\left( {a,b} \right),\left( {b,a} \right) \in O: {\rm{\,\,}}\mathop \downarrow \limits_b^a \phi, \mathop \uparrow \limits_b^a \psi{\rm{\,\,or\,\,}}\mathop \uparrow \limits_b^a \phi, \mathop \downarrow \limits_b^a \psi$;
\item[(2)]
$\forall\,\,a,b {\rm{\,\,such\,\,that\,\,}}\left( {a,b} \right)\in O,\left( {b,a} \right) \notin O: {\rm{\,\,}}\mathop \downarrow \limits_b^a \phi, \mathop \uparrow \limits_b^a \psi{\rm{\,\,or\,\,}}\mathop \uparrow \limits_b^a \phi, \mathop \downarrow \limits_b^a \psi {\rm{\,\,or\,\,}}
\mathop \downarrow \limits_b^a \phi, \mathop \downarrow \limits_b^a \psi$;
\item[(3)]
$\forall\,\,a,b {\rm{\,\,such\,\,that\,\,}}\left( {a,b} \right)\notin O,\left( {b,a} \right) \notin O: {\rm{\,\,}}\phi|_{\left\{ {a,b} \right\}} \left( {x,y} \right) = x, \psi|_{\left\{ {a,b} \right\}} \left( {x,y} \right) = y$.
\end{itemize}
\end{definition}

For a binary operation $f\in F$ define $Com(f) = \left\{\{a,b\}: f|_{\{a,b\}}\rm{\,\,is\,\,commutative}\right\}$.
Consider any binary operation $f_{max}\in F$ which has maximal set $Com(f_{max})$,
i.e. there is no $f\in F$ such that $Com(f_{max})\subset Com(f)$. Since for any binary operations $a,b$,
$Com(a\left(x,y\right))\cup Com(b\left(x,y\right)) = Com(a\left(b\left(x,y\right),b\left(y,x\right)\right))$, we conclude that
$Com(f_{max}) = \bigcup\limits_{f\in F} {Com(f)}$. Define $Com^{o}(f_{max}) = \left\{ {\left( {a,b} \right)|\left\{ {a,b} \right\}\in Com(f_{max})} \right\}$.

\begin{theorem}
\label{weak}
If $F$ satisfies the necessary local conditions for $R$, $UG\left( R \right)$ is complete and $T_F^R = \left( {M^o\cap E\left( R \right) ,P} \right)$ is bipartite,
then there is a pair $\phi, \psi \in F$ which is a weak tournament pair on $Com^{o}(f_{max}) \cap E\left( R \right)$.
\end{theorem}

\begin{proof}
Let $M_1, M_2$ denote a partitioning of vertices of the bipartite graph $T_F^R$.
Then, for every $\left( {a,b} \right),\left( {c,d} \right) \in M_1  $,
there is a function $\phi\in F :\mathop  \downarrow \limits_b^a \mathop  \downarrow \limits_d^c \phi $.
Let us prove by induction that for every
$\left( {a_1 ,b_1 } \right),\left( {a_2 ,b_2 } \right),\dots,\left( {a_n ,b_n } \right) \in M_1  $,
there is a $\phi :\mathop  \downarrow \limits_{b_1 }^{a_1 } \mathop  \downarrow \limits_{b_2 }^{a_2 } \dots\mathop  \downarrow \limits_{b_n }^{a_n } \phi $.

The base of induction $n = 2$ is obvious.
Let $\left( {a_1 ,b_1 } \right)$, $\left( {a_2 ,b_2 } \right)$, $\dots$, $\left( {a_{n + 1} ,b_{n + 1} } \right) \in M_1  $ be given.
By the induction hypothesis, there are  $
\phi _1 ,\phi _2 ,\phi _3\in F :\mathop  \downarrow \limits_{b_2 }^{a_2 } \dots\mathop  \downarrow \limits_{b_n }^{a_n } \mathop  \downarrow \limits_{b_{n + 1} }^{a_{n + 1} } \phi _1 ,\mathop  \downarrow \limits_{b_1 }^{a_1 } \mathop  \downarrow \limits_{b_3 }^{a_3 } \dots\mathop  \downarrow \limits_{b_n }^{a_n } \mathop  \downarrow \limits_{b_{n + 1} }^{a_{n + 1} } \phi _2 ,\mathop  \downarrow \limits_{b_1 }^{a_1 } \mathop  \downarrow \limits_{b_2 }^{a_2 } \dots\mathop  \downarrow \limits_{b_n }^{a_n } \phi _3$.
Then, it is easy to see that
$\mathop  \downarrow \limits_{b_1 }^{a_1 } \dots\mathop  \downarrow \limits_{b_n }^{a_n } \mathop  \downarrow \limits_{b_{n + 1} }^{a_{n + 1} } \phi _3 \left( {\phi _1 \left( {x,y} \right),\phi _2 \left( {x,y} \right)} \right)$ which completes the induction proof.

The analogous statement can be proved for $M_2$.
So it follows from the proof that there are binary operations
$\phi ',\psi '\in F$, such that
$\forall \left( {x,y} \right) \in M_1 {\rm{: }}\mathop  \downarrow \limits_y^x \phi '$ and $\forall \left( {x,y} \right) \in M_2 {\rm{: }}\mathop  \downarrow \limits_y^x \psi '$.

If $\left( {a,b} \right),\left( {b,a} \right) \in Com^{o}(f_{max}) \cap E\left( R \right)$, then the necessary local conditions for $C$ give that
$\{a,b\}\in M$. Moreover, $\left( {a,b} \right),\left( {b,a} \right) \in M^o \cap E\left( R \right)$
are always in different partitions of $T_F^R$.
Consequently, $\phi ',\psi '$ satisfy the first property of a weak tournament pair on  $Com^{o}(f_{max}) \cap E\left( R \right)$.
Our goal is to construct a pair of operations that satisfy other two properties.

Consider operations $\phi ''\left(x,y\right) = f_{max}\left(\phi '\left(x,y\right),\phi '\left(y,x\right)\right)$ and
$\psi ''\left(x,y\right) = f_{max}\left(\psi '\left(x,y\right),\psi '\left(y,x\right)\right)$.
Since for any $B\in Com(\phi ')$, $\phi ''|_{B} = \phi '|_{B}$ and $\psi ''|_{B} = \psi '|_{B}$, we obtain that $\phi '',\psi ''$ satisfy the
first property of a weak tournament pair on  $Com^{o}(f_{max}) \cap E\left( R \right)$, too.
Moreover, $\phi '',\psi ''$ are commutative on $Com(f_{max})$.
For any two elements $a,b\in A$ such that $\left( {a,b} \right) \in Com^{o}(f_{max})\cap E\left( R \right), \left( {b,a} \right)\notin Com^{o}(f_{max})\cap E\left( R \right)$,
$\phi ''|_{\left\{ {a,b} \right\}},\psi ''|_{\left\{ {a,b} \right\}}$ are both commutative. If $(a,b)$ is a vertex in $T_F^R$, then one of
$\phi ',\psi '$ satisfies $\mathop \downarrow \limits_b^a$, and therefore it holds for one of $\phi '',\psi ''$, too. If $(a,b)$ is not a vertex
in $T_F^R$, i.e. $\{a,b\}\notin M$, then $\mathop \downarrow \limits_b^a\phi '', \mathop \downarrow \limits_b^a\psi ''$, because $F|_{\{a,b\}}$ contains only one commutative operation.
In both cases the second property of a weak tournament pair on  $Com^{o}(f_{max})\cap E\left( R \right)$ is satisfied.

In case that $(a,b), (b,a)\notin Com^{o}(f_{max})\cap E\left( R \right)$, we see that $\{a,b\}\notin Com(f_{max})$ (the case when $\{a,b\}\in Com(f_{max}), (a,b), (b,a)\notin E\left( R \right)$ is impossible due to completeness of $UG\left( R \right)$) which means that $\phi ''_{\left\{ {a,b} \right\}},\psi ''_{\left\{ {a,b} \right\}}$ are projections. Thus, a pair of operations $\phi \left( {x,y} \right) = \phi ''\left( {x,\phi ''\left( {y,x} \right)} \right)$, $\psi \left( {x,y} \right) = \psi ''\left( {\psi ''\left( {y,x} \right),y} \right)$
satisfy all three properties of  a weak tournament pair on $Com^{o}(f_{max}) \cap E\left( R \right)$.
\end{proof}

\begin{theorem}
\label{arithmetical}
If $F$ satisfies the necessary local conditions for $R$ and
$\overline {Com(f_{max})}  \ne \emptyset $, then $F$ contains an arithmetical operation on $\overline {Com(f_{max})}$.
\end{theorem}

\begin{proof}
Obviously, for every $B \in \overline{Com(f_{max})} $, $F|_B $ cannot contain any commutative binary function.
Therefore, every binary function in $F|_B $ is a projection.

For $B \in \overline{Com(f_{max})}$, let $m^B$ be an arithmetical function on $B$; existence of this function follows from the necessary local conditions for $R$.
Assume now that $\overline{Com(f_{max})}  = \left\{ {\left\{ {x_1 ,y_1 } \right\},\dots,\left\{ {x_s ,y_s } \right\}} \right\}$.
We prove by induction that for every $r \le s$, $F$ contains a function $m_r :A^3  \to A$ that is arithmetical on
$\left\{ \left\{ {x_i ,y_i } \right\} | 1\leq i \leq r\right\}$.

When $r = 1$, $m_1 \left( {x,y,z} \right) = m^{\left\{ {x_1 ,y_1 } \right\}} \left( {x,y,z} \right)$ and the statement is
obviously true.
Suppose it is true for $r \le k < s$ and that we have the function $m_k :A^3  \to A$.
Let us prove the statement for $r = k + 1$. If $m_k$ is arithmetical on $\left\{\left\{ {x_{k + 1} ,y_{k + 1} } \right\}\right\}$, then we define $m_{k + 1}  \buildrel \Delta \over = m_k $
and the statement is proved. Otherwise, one of the following three statements is true
\[
\exists x,y \in \left\{ {x_{k + 1} ,y_{k + 1} } \right\}{\rm{ }}\left[ {m_k \left( {x,x,y} \right) \ne y} \right],
\]
\[
\exists x,y \in \left\{ {x_{k + 1} ,y_{k + 1} } \right\}{\rm{ }}\left[ {m_k \left( {y,x,x} \right) \ne y} \right],
\]
\[
\exists x,y \in \left\{ {x_{k + 1} ,y_{k + 1} } \right\}{\rm{ }}\left[ {m_k \left( {y,x,y} \right) \ne y} \right].
\]

Suppose the first case holds (the proof for other cases is analogous), i.e.
$m_k |_{\left\{ {x_{k + 1} ,y_{k + 1} } \right\}} \left( {x,x,y} \right)$ is the $x$-projection.
It is easy to see that the function $m_{k + 1} \left( {x,y,z} \right) = m_k \left( {m^{\left\{ {x_{k + 1} ,y_{k + 1} } \right\}} \left( {x,y,z} \right),m^{\left\{ {x_{k + 1} ,y_{k + 1} } \right\}} \left( {x,y,z} \right),m_k \left( {x,y,z} \right)} \right)$
is arithmetical on $\left\{ \left\{ {x_i ,y_i } \right\} | 1\leq i \leq k+1\right\}$.

Induction completed and it is clear that $m_s \left( {x,y,z} \right)$ satisfies the condition of theorem.
\end{proof}

\section{Proof of Theorem \ref{main}}
Suppose we have some instance $\left(V, {\sf C}, W\right)$ of $MinHom_{R}\left(Inv\left(F\right)\right)$. Let $Sol\left(V, {\sf C}, W\right)$ denote the set of its solutions.
By Theorem \ref{weak}, there is a pair $\phi, \psi \in F$ which is a weak tournament pair on $Com^{o}(f_{max}) \cap E\left( R \right)$.

Until the end of this section, the pair $\phi, \psi$ will be fixed.
For any $f,g:V\rightarrow A$, define $\phi\left(f,g\right),\psi\left(f,g\right):V\rightarrow A$ such that $\phi\left(f,g\right)(v) = \phi\left(f(v),g(v)\right)$, $\psi\left(f,g\right)(v) = \psi\left(f(v),g(v)\right)$. Since any constraint relation of ${\sf C}$ is preserved by $\phi, \psi$, we have $\phi\left(f,g\right),\psi\left(f,g\right)\in Sol\left(V, {\sf C}, W\right)$ for $f,g\in Sol\left(V, {\sf C}, W\right)$.

\begin{lemma}
\label{multimorpism}
For any measure $M(f) = \sum\limits_{v\in V} \sum\limits_{r\in R}{w_{vr}r\left( f\left( v \right) \right)}$ we have
\[
M(\phi\left(f,g\right)) + M(\psi\left(f,g\right)) \leq M(f) + M(g)
\]
\end{lemma}

\begin{proof}
Let us prove that for any $a,b\in A$ and any $r\in R$: $r\left(a\right)+r\left(b\right) \geq r\left(\phi\left(a,b\right)\right)+r\left(\psi\left(a,b\right)\right)$.

If $\left( {a,b} \right),\left( {b,a} \right) \in Com(f_{max})\cap E\left( R \right)$ or $\left\{ {a,b} \right\} \notin Com(f_{max})$,
then $\left\{\phi\left(a,b\right), \psi\left(a,b\right)\right\} = \left\{a, b\right\}$, and we have an equality.
If $\left( {a,b} \right)\in Com(f_{max}) \cap E\left( R \right),\left( {b,a} \right) \notin E\left( R \right)$,
then $\left\{\phi\left(a,b\right), \psi\left(a,b\right)\right\} = \left\{a, b\right\}$ or $\left\{\phi\left(a,b\right), \psi\left(a,b\right)\right\} = \left\{b\right\}$. In the first case we have an equality and in the second case we have:
\[
\forall\,\,r\in R,\,\,r\left(a\right)+r\left(b\right) \geq 2r\left(b\right)\Leftrightarrow \forall\,\,r\in C\,\,r\left(b\right) \leq r\left(a\right) \Leftrightarrow (b,a)\notin E\left( R \right)
\]
From this inequality, we conclude that for any measure $M(f) = \sum\limits_{v\in V}^n \sum\limits_{r\in R}{w_{vr}r\left( f\left( v \right) \right)}$ we have
$
M(f) + M(g) \geq M(\phi\left(f,g\right)) + M(\psi\left(f,g\right))
$
\end{proof}

\begin{lemma}
\label{delete}
Suppose $A_v = \left\{f(v)| f\in Sol\left(V, {\sf C}, W\right)\right\}$ and there is $a,b\in A_v$ such that $\mathop \downarrow \limits_b^a \psi$, $\mathop \downarrow \limits_b^a \phi$. Then we have
\[
\min \limits_{f\in Sol\left(V, {\sf C}, W\right), f(v)=b} {M(f)} \leq \min \limits_{f\in Sol\left(V, {\sf C}, W\right), f(v)=a} {M(f)}
\]
\end{lemma}

\begin{proof}
Suppose that $f^{*}\in Sol\left(V, {\sf C}, W\right), f^{*}(v)=a$ and $f^{**}\in Sol\left(V, {\sf C}, W\right), f^{**}(v)=b$ such that
\[
M(f^{*}) = \min \limits_{f\in Sol\left(V, {\sf C}, W\right), f(v)=a} {M(f)},
M(f^{**}) = \min \limits_{f\in Sol\left(V, {\sf C}, W\right), f(v)=b} {M(f)}
\]
Since $\phi\left(f^{*},f^{**}\right)(v) = b$, $\psi\left(f^{*},f^{**}\right)(v) = b$, then by Lemma \ref{multimorpism}

$2M(f^{**}) \leq M(\phi\left(f^{*},f^{**}\right)) + M(\psi\left(f^{*},f^{**}\right)) \leq M(f^{*}) + M(f^{**})$
\qed
\end{proof}

Now we are ready to describe polynomial-time algorithm for $MinHom\left(Inv\left(F\right)\right)$, which conceptually follows the proof of Theorem 8.3 in \cite{cohen}. First we compute sets  $A_v = \left\{f(v)| f\in Sol\left(V, {\sf C}, W\right)\right\}$ for every $v\in V$. This could be done polynomial-time, because, by result of \cite{bulatov}, $CSP\left(Inv\left(F\right)\right)$ is polynomial-time solvable by 3-consistency algorithm: adding a constraint $(v,\{a\})$ to $\sf{C}$ we can find whether $a\in A_v$ or not.
After this operation, we iteratively delete all elements $a$ from every $A_v$ such that there is $b\in A_v: \mathop \downarrow \limits_b^a \psi, \mathop \downarrow \limits_b^a \phi$. This operation will not increase a minimum of optimized functional due to Lemma \ref{delete}.
Afterwards, we update every constraint pair $\left(\left(v_1,\dots, v_m\right), \rho\right)\in \sf{C}$ by defining $\rho = \rho \cap A_{v_1}\times \dots \times A_{v_m}$.

Consider a collection of sets ${\sf A} = \left\{A_v| v\in V\right\}$. Obviously, $|{\sf A}|< 2^{|A|}$. For a pair $\left(\left(v_1,\dots, v_m\right), \rho\right)\in\sf{C}$, a predicate $\rho$ can be considered as multi-sorted predicate over ${\sf A}$ with signature $\left(A_{v_1},\dots, A_{v_m}\right)$. Moreover, for any $p\in F$ this multi-sorted predicate preserves a polymorphism $\left\{p^{B}| B\in \sf{A}\right\}$ where $p^B = p|_B$.

By Theorem \ref{arithmetical}, $F$ contains an operation $m$ which is arithmetical on $\overline {Com(f_{max})}$.
Consider multi-sorted polymorphisms $\left\{\phi^{B}| B\in \sf{A}\right\}, \left\{\psi^{B}| B\in \sf{A}\right\}, \left\{m^{B}| B\in \sf{A}\right\}$. It is easy to see that this polymorphisms satisfy conditions of Theorem \ref{multi} and the set of multi-sorted predicates $\left\{\rho| \left(s, \rho\right)\in \sf{C}\right\}$ is tractable. Therefore, we can polynomially solve our problem.

\section{Directions for further work}
As it was said in the introduction, $MinHom_R$ fits the framework of Valued CSP ($VCSP$) \cite{cohen}. The following subproblems of $VCSP$ generalize $MinHom$ and are not investigated yet.

1. What happens with $MinHom_R$ when $UG\left( R \right)$ is not complete. We believe that Theorem \ref{main} is true for this case too.

2. Suppose we have a finite valued constraint language $\Gamma$, i.e. a set of valued predicates over some finite domain set. If $\Gamma$ contains all unary valued predicates, we call $VCSP(\Gamma)$ a conservative $VCSP$. This name is motivated by the fact that in this case the multimorphisms (which is a generalization of polymorphisms for valued constraint languages \cite{cohen}) of $\Gamma$ must consist of conservative functions. Since there is a well-known dichotomy for conservative CSPs \cite{bulatov}, we suspect that there is a dichotomy for conservative $VCSPs$.

%
%

%
\end{document}